 \documentclass[final,3p,times]{elsarticle}

\usepackage{amssymb}
\usepackage{amsthm}

\usepackage{graphicx}
\usepackage{amsmath}
\usepackage{amsfonts}
\usepackage{subfig}
\usepackage{algpseudocode}
\usepackage[noline,ruled]{algorithm2e}
\usepackage{multicol}
\usepackage{multirow}
\usepackage{units}
\usepackage[usenames,dvipsnames]{xcolor}
\usepackage{epstopdf}

\usepackage{hyperref}
\hypersetup{
    colorlinks,
    citecolor=black,
    filecolor=black,
    linkcolor=black,
    urlcolor=black
}
\newtheorem{theorem}{Theorem}
\newtheorem{definition}{Definition}
\newtheorem{lemma}{Lemma}

\newcommand{\class}{$C^{(h)}$}

\makeatletter
\def\ps@pprintTitle{%
 \let\@oddhead\@empty
 \let\@evenhead\@empty
 \def\@oddfoot{\centerline{\thepage}}%
 \let\@evenfoot\@oddfoot}
 \makeatother

\begin{document}

\begin{frontmatter}



\title{Improved Lower Bounds for Online Hypercube and Rectangle Packing}


\author[db]{David Blitz} 
\ead{d.c.blitz@robeco.nl}
\author[sh]{Sandy Heydrich\fnref{fn}} 
\ead{heydrich@mpi-inf.mpg.de}
\author[rvs]{Rob van Stee}
\ead{rob.vanstee@uni-siegen.de}
\author[avv]{Andr\'e van Vliet}
\ead{andrevanvliet@xs4all.nl}
\author[gw]{Gerhard J. Woeginger}
\ead{woeginger@algo.rwth-aachen.de}

\fntext[fn]{Sandy Heydrich is supported by a Google Europe PhD Fellowship.}

\address[db]{Robeco, Weena 850, 3014 DA Rotterdam, The Netherlands}
\address[sh]{Max Planck Institute for Informatics, Building E1.4, and Saarbr\"{u}cken Graduate School of Computer Science, Building E1.3, Saarland Informatics Campus, D-66123 Saarbr\"ucken, Germany}
\address[rvs]{Department of Mathematics, Universit\"{a}t Siegen, Walter-Flex-Strasse 3, D-57068 Siegen, Germany}
\address[avv]{Zwaluwlaan 18, 2261 BR Leidschendam, The Netherlands}
\address[gw]{Department of Computer Science, RWTH Aachen,
D-52056 Aachen, Germany}

\begin{abstract}
Packing a given sequence of items into as few bins as possible in an online fashion is a widely studied problem. 
We improve lower bounds for packing boxes into bins in two or more dimensions, both for general algorithms for squares and rectangles (in two dimensions) and for an important subclass, so-called Harmonic-type algorithms for hypercubes (in two or more dimensions). Lastly, we show that two adaptions of ideas from a one-dimensional packing algorithm \cite{HS15} to square packing do not help to break the barrier of 2.
\end{abstract}

\end{frontmatter}

\section{Introduction}

In this paper, we consider the problem of online bin packing in two or more dimensions. 
This problem is defined as follows: We receive a sequence of boxes $h_1, \ldots , h_n$ (called \textit{items}) in $d$-dimensional space, and each item $h_i$ has size $s_i^{(j)}$ in dimension $j$. We also consider the special case where all items are hypercubes with the same edge length $s_i$ in all dimensions. 
We furthermore have an infinite number of bins, which are hypercubes of edge length one. 
We have to assign each item $h_i$ to a bin and a position $(x_1, \ldots, x_d)$ inside this bin, such that $0 \le x_j \le 1-s_i^{(j)}$ for all $1 \le j \le d$ and no two items in the same bin are overlapping. Items must be placed parallel to the axes of the bins and in the orientation given in the input (i.e., rotations are not allowed).
We call a bin \textit{used} if at least one item is assigned to it, and our goal is to minimize the number of used bins. 
The online setting requires us to assign an item to a bin immediately when it arrives, without knowledge of future items. 
We consider this problem in two or more dimensions.

For measuring the quality of a solution of the algorithm, we use the standard notion of \textit{asymptotic performance ratio}. 
For an input sequence $\sigma$, let $\mathcal{A}(\sigma)$ be the number of bins algorithm $\mathcal{A}$ uses to pack the items in $\sigma$ and let $OPT(\sigma)$ be the minimum number of bins in which these items can be packed. 
The asymptotic performance ratio for $A$ is defined as
\begin{align*}
R_\mathcal{A}^\infty = \limsup_{n \rightarrow \infty}\sup_{\sigma}\left\{ \left.\frac{\mathcal{A}(\sigma)}{OPT(\sigma)} \right\vert OPT(\sigma)=n \right\}
\end{align*}
If $\mathcal{O}$ denotes a class of packing algorithms, then the optimal asymptotic performance ratio for class $\mathcal{O}$ is defined as $R_\mathcal{O}^\infty = \inf_{\mathcal{A} \in \mathcal{O}}R_\mathcal{A}^\infty$. From now on, we will only talk about asymptotic performance ratios, although we omit the word asymptotic.

\subsection{Previous Results}
The classic online bin packing problem in one dimension was first considered by Ullman \cite{Ullman71}, and he also gave the \textsc{FirstFit} algorithm with performance ratio $\frac{17}{10}$ \cite{JDUGG74}. 
The \textsc{NextFit} algorithm was introduced by Johnson \cite{Johnson74}, who showed that this algorithm has a performance ratio of 2.

The \textsc{Harmonic} algorithm was introduced by Lee and Lee \cite{LL85}. 
If we define $u_1=2, u_{i+1}=u_i(u_i-1)+1$, then this algorithm has performance ratio $h_\infty = \sum_{i=1}^{\infty}\frac{1}{u_i-1} < 1.69104$. 
It uses bounded space (i.e. only a constant number of bins are open at a time, meaning that items can be added to them) and they showed that no algorithm with this property can perform better. 
Later, various improvements of this approach were given (using unbounded space), including \textsc{RefinedHarmonic} (performance ratio $\frac{373}{228} < 1.63597$) \cite{LL85}, 
\textsc{ModifiedHarmonic} (performance ratio $< 1.61562$) and \textsc{ModifiedHarmonic2} (performance ratio $<1.61217$) by Ramanan et al. \cite{RBLL89}, \textsc{Harmonic++} (performance ratio $<1.58889$) by Seiden \cite{Seiden02}, and finally \textsc{SonOfHarmonic} (performance ratio $1.5816$) by Heydrich and van Stee \cite{HS15}. Very recently, the bound has further been improved to $1.5783$ by the algorithm \textsc{AdvancedHarmonic} \cite{BBDEL17a}. The best general lower bound of 1.54037 for online bin packing in one dimension was given by Balogh et al. \cite{BBG11}.

Online bin packing of rectangles was first discussed by Coppersmith and Raghavan \cite{CR89}. 
They gave an algorithm which has in two dimensions a performance ratio of $\frac{13}{4}$ for general rectangles and $\frac{43}{16}$ for squares. 
Additionally, they showed a lower bound of $\frac{4}{3}$ for square packing in any dimension $d \ge 2$. 
Csirik and van Vliet improved upon this by giving an algorithm that achieves $h_\infty^d$ performance ratio for any dimension $d \ge 2$ \cite{CV93}. 
They also show that this is a lower bound for bounded space algorithms, although their algorithm uses unbounded space. 
Later, Epstein and van Stee provided a bounded space algorithm that matches this lower bound \cite{ES05}. 
In the same paper, they also give an optimal online bounded space algorithm for box packing (i.e. items are not hypercubes anymore but can have different sizes in different dimensions), although they do not provide the exact performance ratio. 
Finally, Han et al. \cite{HCTZZ11} gave an upper bound of $2.5545$ for the special case of $d=2$, which is the best bound currently known.

The best known lower bounds for hypercube packing are $1.6406$ for two dimensions and $1.6680$ for three dimensions \cite{ES05b}. For box packing, the best known lower bounds are $1.851$ for two dimensions and $2.043$ for three dimensions \cite{vanVlietPhD}.
Regarding upper bounds, the best algorithm for square packing achieves a performance ratio of $2.1187$ and the best algorithm for cube packing achieves $2.6161$~\cite{HYZ10}. For rectangle packing, a 2.5545-competitive algorithm is known, as well as a 4.3198-competitive algorithm for online three dimensional box packing~\cite{HCTZZ11}.

\subsection{Our Contribution}

We improve the general lower bound for square packing in two dimensions to $1.680783$. 
In the upcoming WAOA 2017, Epstein et al. improved this lower bound further to $1.75$, using different methods \cite{BBDEL17b}.
For rectangle packing, we improve the general lower bound to $1.859$.
%
Furthermore, we improve the lower bound for Harmonic-type algorithms for hypercube packing in any dimension $d \ge 2$. 
This uses a generalization of the method of Ramanan et al. \cite{RBLL89}. 
In particular, we show that such an algorithm cannot break the barrier of $2$ for $d = 2$, by giving a lower bound of $2.02$ for this case. 
This shows that substantially new ideas will be needed in order to improve significantly on the current best upper bound
of $2.1187$ and get close to the general lower bound. 
Our lower bound tends to $3$ for large numbers of dimensions.

Lastly, we also show that even when incorporating two central ideas from the currently best one-dimensional bin packing algorithm \cite{HS15} into two-dimensional square packing, similar lower bounds as those for Harmonic-type algorithms can still be achieved.

\subsection{Preliminaries}

At several points in this paper, we use the notion of \emph{anchor points} as defined by Epstein and van Stee \cite{ES07}. 
We assign the coordinate $(0, \ldots, 0)$ to one corner of the bin, all edges connected to this corner are along a positive axis and have length 1. 
Placing an item at an anchor point means placing this item parallel to the axes such that one of its corners coincides with the anchor point and no point inside the item has a smaller coordinate than the corresponding coordinate of the anchor point. 
We call an anchor point \emph{blocked} for type $s$ items in a certain packing (i.e. in a bin that contains some items), if we cannot place an item of type $s$ at that anchor point (without overlapping other items).

\section{Lower Bound for General Algorithms for Square Packing}

\subsection{Van Vliet's Method}

For deriving a general lower bound on the performance ratio of online hypercube packing algorithms, we extend an approach by van Vliet \cite{vanVlietPhD} based on linear programming. 
Problem instances considered in this approach are characterized by a list of items $L=L_1\ldots L_k$ for some $k \ge 2$, where each sublist $L_j$ contains $\alpha_j \cdot n$ items of side length $s_j$ (we will also call such items ``items of size $s_j$'' or simply ``$s_j$-items''). 
We assume $s_1 \le \ldots \le s_k$. 
The input might stop after some sublist. 
An online algorithm $\mathcal{A}$ does not know beforehand at which point the input sequence stops, and hence the asymptotic performance ratio can be lower bounded by $$R \ge \min_\mathcal{A}\max_{j=1, \ldots, k}\limsup_{n\rightarrow\infty}\frac{\mathcal{A}(L_1, \ldots, L_j)}{OPT(L_1, \ldots, L_j)}$$

For this approach, we define the notion of a \textit{pattern}: A pattern is a multiset of items that fits in one bin. 
We denote a pattern by a tuple $(p_1, \ldots, p_k)$, where $p_i$ denotes the number of $s_i$-items contained in the pattern (possibly zero). 
We call a pattern $p$ \emph{dominant} if the multiset consisting of the items of $p$ and an additional item of the smallest item size that is used by $p$ cannot be packed in one bin. 
The performance of an online algorithm on the problem instances we consider can be characterized by the number of bins it packs according to a certain pattern. 
Van Vliet denotes the set of all feasible patterns by $T$, which is the union of the disjoint sets $T_1, \ldots, T_k$ where $T_j$ contains patterns whose first non-zero component is $j$ (i.e., whose smallest item size used is $s_j$). 
We can then calculate the cost of an algorithm $\mathcal{A}$ by $\mathcal{A}(L_1, \ldots, L_j)=\sum_{i=1}^{j}\sum_{p \in T_i}n(p)$, where $n(p)$ denotes the number of bins $\mathcal{A}$ packs according to pattern $p$. 
Note that we only need to consider dominant patterns in the LP \cite{vanVlietPhD}. 
As the variables $n(p)$ characterize algorithm $\mathcal{A}$, optimizing over these variables allows us to minimize the performance ratio over all online algorithms with the following LP:

\begin{equation*}
\begin{array}{ll@{}ll}
\text{minimize} & \,\, R\\
\text{subject to} & \hspace{-1mm} \displaystyle\,\,\,\sum_{p \in T}p_j \cdot x(p) \ge \alpha_j & \,\,\,\, 1 \le j \le k\\
& \, \, \displaystyle \sum_{i=1}^{j}\sum_{p \in T_i}  x(p) \le \lim\limits_{n \rightarrow \infty}\frac{OPT(L_1, \ldots, L_j)}{n}R & \,\,\,\, 1 \le j \le k\\
& x(p) \ge 0 & \,\,\,\,\forall p \in T
\end{array}
\end{equation*}

In this LP, the variables $x(p)$ replace $n(p)/n$, as we are only interested in results for $n\rightarrow\infty$.
Note that item sizes are always given in nondecreasing order to the algorithm. 
In this paper, however, we will often consider item sizes in \emph{nonincreasing} order for constructing the input sequence and generating all patterns.

\subsection{Proving a lower bound of $1.680783$}

In this section, we will prove the following theorem:

\begin{theorem}\label{thm:general_lb}
No online algorithm can achieve a competitive ratio of $1.680783$ for the online square packing problem.
\end{theorem}

\begin{table}
\centering
\begin{tabular}{|c|c|c|c|}
\hline 
sublist $L_i$ & number of items $\alpha_i$ & item size $s_i$ & $OPT(L_1\ldots L_i)\cdot 176400/n$\\ 
\hline 
$L_1$ & $839n$ & $1/420-\epsilon$ & $839$ \\ 
\hline 
$L_2$ & $10n$ & $1/105+\epsilon/105$ & $999$ \\ 
\hline 
$L_3$ & $8n$ & $1/84+\epsilon/84$ & $1199$ \\ 
\hline 
$L_4$ & $4n$ & $1/42+\epsilon/42$ & $1599$ \\ 
\hline 
$L_5$ & $39n$ & $1/21+\epsilon/21$ & $17199$ \\ 
\hline 
$L_6$ & $8n$ & $1/20+\epsilon/20$ & $20727$ \\ 
\hline 
$L_7$ & $4n$ & $1/10+\epsilon/10$ & $27783$ \\ 
\hline 
$L_8$ & $7n$ & $1/5+\epsilon/5$ & $77175$ \\ 
\hline 
$L_9$ & $5n$ & $1/4+\epsilon/4$ & $132300$ \\ 
\hline 
$L_{10}$ & $n$ & $1/2+\epsilon/2$ & $176400$ \\ 
\hline 
\end{tabular} 
\caption{The input sequence that gives a lower bound of $1.680783$ together with optimal solutions.}
\label{tab:input}
\end{table}

Consider the input sequence in Table \ref{tab:input}.
First of all, we need to prove the correctness of the values $\frac{OPT(L_1\ldots L_j)}{n}$ for $j=1,\ldots,k$ and $n\to\infty$. 
To prove a lower bound, we do not need to prove optimality of the offline packings that we use. It is sufficient to prove feasibility.
To do this, we use anchor packings. An anchor packing is a packing where every item is placed at an anchor point. In this section, we use $420^2$ anchor points. The anchor points are at the positions for which both coordinates are integer multiples of $(1+\epsilon)/420$.
Note that every item used in the construction apart from the ones in $L_1$ have sides which are exact multiples of $(1+\epsilon)/420$. Therefore, whenever we place an item at an anchor point, and the item is completely
contained within the bin, it will fill exactly a square bounded by anchor points on all sides.

To check whether a given pattern is feasible, the items of size $s_1$ can be considered separately. Having placed all other items at anchor points, we can place exactly one item of size $s_1$ at each anchor point which is still available. Here an anchor point $(x,y)$ is available if no item covers the point $(x+\epsilon,y+\epsilon)$. By the above,
after all other items have been placed at anchor points, it is trivial to calculate the number of available anchor points; at least all the anchor points with at least one coordinate equal to $(1+\epsilon)419/420$ are still available.

For any pattern that we use, the largest items in it are always arranged in a square grid at the left bottom corner of the bin (at anchor points). The second largest items are arranged in an L-shape around that square. It is straightforward to calculate the numbers of these items as well. 
The patterns used for the given upper bounds on the optimal solution are listed in Table \ref{tab:optimal-patterns}. Note that not all of these patterns are greedy (in the sense that we add, from larger to smaller items, always as many items of the current type as still fit).

Let us give some intuition on how these patterns are constructed. We start by finding a pattern that contains the maximal number of the largest type of items, and then add greedily as many items as possible of the second largest type, then third largest type and so on. We take as many bins with this pattern as are necessary to pack all the largest items; a certain number of items of all other types remain. We continue by choosing the pattern that contains the largest possible number of items of the second-largest type and fill it up greedily as before with other items. We use this pattern in such a number of bins that all remaining items of the second-largest type are packed. We continue like that until all items are packed.
You can see this approach for example in the patterns used for $OPT(L_1\ldots L_3)$. We can pack 6889 items of size $s_3$ into one bin. With these, we can pack no more than 207 $s_2$-items, and finally we can add at most $863$ $s_1$-items; this gives the first pattern. We need $n\cdot 8/6889$ bins with this pattern to pack the $8n$ $s_3$-items. This leaves $n \cdot 67234/6889$ items of size $s_2$ unpacked, and as we can pack at most $10816$ $s_2$-items into one bin (and $3344$ $s_1$-items with them), this gives a certain amount of bins with this second pattern $(3344,10816,0,\ldots,0)$.

However, we sometimes slightly derive from this construction, e.g., in the patterns used for $OPT(L_1\ldots L_5)$. In a bin with 400 items of size $s_5$, we could fit 81 items of size $s_4$. However, if we do so, we would pack more $s_4$-items than necessary and thus lose space that we need in order to pack other items. In that case, we reduce the number of $s_4$-items as much as possible while still packing all of them (in this case, we reduce it to 42).

In Figure \ref{fig:packing_p3}, we give the optimal packing for the whole input sequence (i.e., for $L_1\ldots L_{10}$).

\begin{figure}
\centering
\includegraphics[width=.9\textwidth]{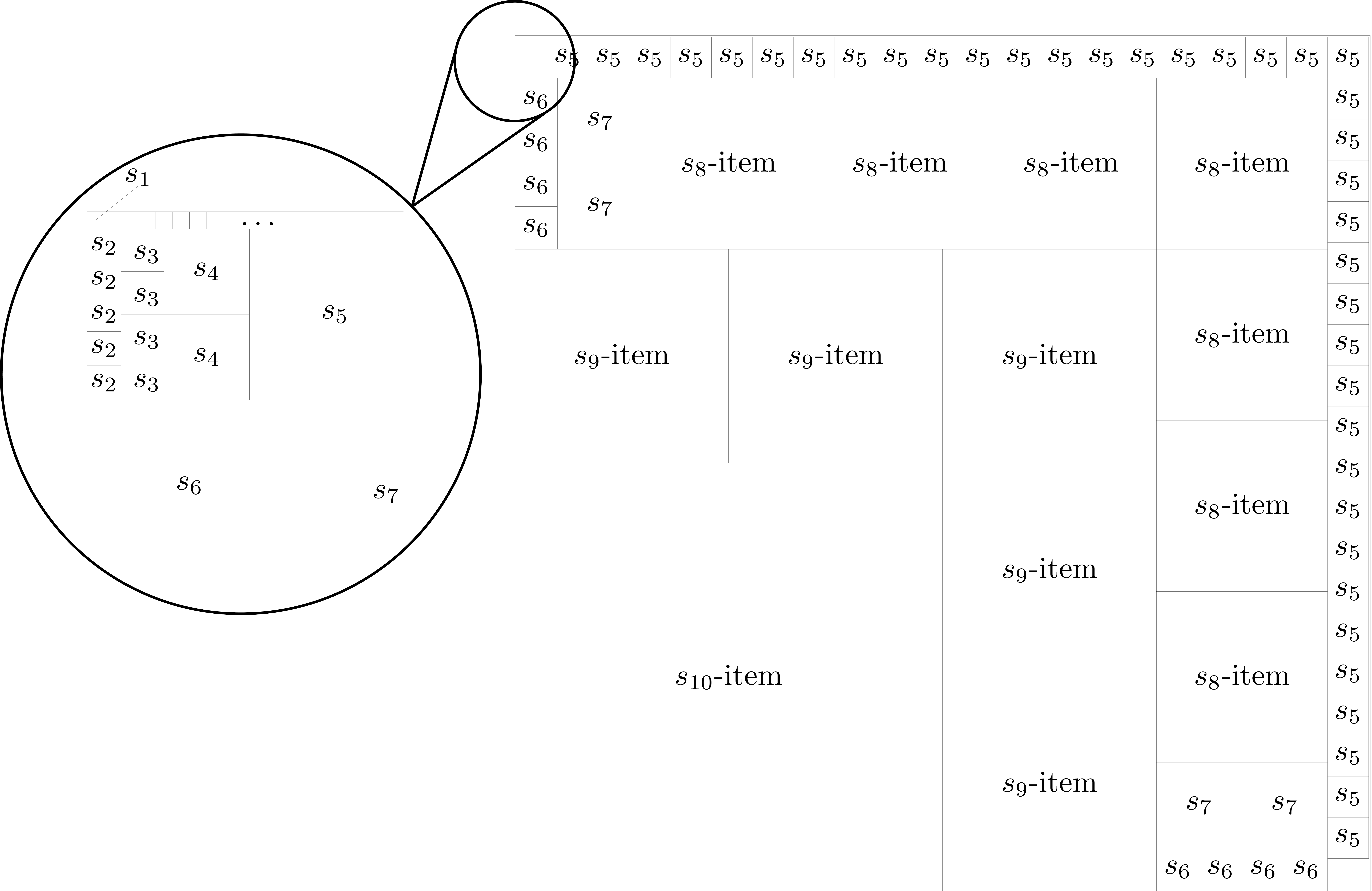}
\caption{How to pack the whole input sequence $L_1\ldots L_{10}$. Note that these sketches are not true to scale for the sake of readability.}\label{fig:packing_p3}
\end{figure}

\begin{table}
\centering
\begin{tabular}{|c|c|c|c|c|c|c|c|c|c|c|c|}
\hline
\multirow{2}{*}{$i$} & \multicolumn{10}{c|}{patterns used for $OPT(L_1\ldots L_i)/n$} & number of bins with this pattern\\
\cline{2-11}
& $s_1$ & $s_2$ & $s_3$ & $s_4$ & $s_5$ & $s_6$ & $s_7$ & $s_8$ & $s_9$ & $s_{10}$ & divided by $n$\\
\hline
$1$ & $176400$&$0$&$0$&$0$&$0$&$0$&$0$&$0$&$0$&$0$ & $ 839/176400$\\
\hline
\multirow{2}{*}{$2$} & $3344$&$10816$&$0$&$0$&$0$&$0$&$0$&$0$&$0$&$0$ & $ 10/10816$\\
\cline{2-11}
& $176400$&$0$&$0$&$0$&$0$&$0$&$0$&$0$&$0$&$0$ & $ 31393/6624800 $\\
\hline
\multirow{3}{*}{$3$} & $863$&$207$&$6889$&$0$&$0$&$0$&$0$&$0$&$0$&$0$ & $ 8/6889$\\
\cline{2-11}
& $3344$&$10816$&$0$&$0$&$0$&$0$&$0$&$0$&$0$&$0$ & $ 33617/37255712$\\
\cline{2-11}
& $176400$&$0$&$0$&$0$&$0$&$0$&$0$&$0$&$0$&$0$ & $ 1944236893/410744224800$\\
\hline
\multirow{4}{*}{$4$} & $863$&$207$&$165$&$1681$&$0$&$0$&$0$&$0$&$0$&$0$ & $ 4/1681$\\
\cline{2-11}
& $863$&$207$&$6889$&$0$&$0$&$0$&$0$&$0$&$0$&$0$ & $ 12788/11580409$\\
\cline{2-11}
& $3344$&$10816$&$0$&$0$&$0$&$0$&$0$&$0$&$0$&$0$ & $ 31961/37255712$\\
\cline{2-11}
& $176400$&$0$&$0$&$0$&$0$&$0$&$0$&$0$&$0$&$0$ & $ 646638631/136914741600$\\
\hline
$5$ & $10477$&$103$&$83$&$42$&$400$&$0$&$0$&$0$&$0$&$0$ & $ 39/400$\\
\hline
\multirow{2}{*}{$6$} & $839$&$16$&$8$&$4$&$39$&$361$&$0$&$0$&$0$&$0$ & $ 8/361$\\
\cline{2-11}
& $10493$&$102$&$83$&$42$&$400$&$0$&$0$&$0$&$0$&$0$ & $13767/144400$\\
\hline
\multirow{3}{*}{$7$} & $839$&$16$&$8$&$4$&$39$&$37$&$81$&$0$&$0$&$0$ & $ 4/81$\\
\cline{2-11}
& $839$&$16$&$8$&$4$&$39$&$361$&$0$&$0$&$0$&$0$ & $ 500/29241$\\
\cline{2-11}
& $10541$&$99$&$83$&$42$&$400$&$0$&$0$&$0$&$0$&$0$ & $ 13143/144400$\\
\hline
$8$ & $1918$&$23$&$18$&$10$&$90$&$19$&$10$&$16$&$0$&$0$ & $ 7/16$\\
\hline
\multirow{2}{*}{$9$} & $839$&$10$&$8$&$4$&$39$&$8$&$4$&$7$&$9$&$0$ & $5/9$\\
\cline{2-11}
& $1918$&$23$&$19$&$10$&$90$&$19$&$10$&$16$&$0$&$0$ & $ 7/36$\\
\hline
$10$ & $839$&$10$&$8$&$4$&$39$&$8$&$4$&$7$&$5$&$1$ & $1$\\
\hline
\end{tabular}
\caption{Patterns used for the optimal solutions.}
\label{tab:optimal-patterns}
\end{table}

In order to prove lower bounds, we will use the dual of the LP given above. It is defined as follows:

\begin{equation*}
\begin{array}{ll@{}ll}
\text{maximize} & \,\, \sum_{j=1}^{k}\alpha_j\lambda_j\\
\text{subject to} & \hspace{-1mm} \displaystyle\,\,\,\sum_{i=j}^{k}\lambda_ip_i + \sum_{i=j}^{k}\mu_i \le 0 & \,\,\,\, \forall p \in T_j, 1\le j \le k\\
& \, \, \displaystyle -\sum_{j=1}^{k}\mu_j \cdot \lim_{n \to \infty} \frac{OPT(L1\ldots L_j)}{n} \le 1 & \\
& \lambda_j \ge 0 & \,\,\,\, 1 \le j \le k\\
& \mu_j \le 0 & \,\,\,\, 1 \le j \le k
\end{array}
\end{equation*}

Note that any feasible solution to this dual gives us a valid lower bound for the problem. In Table \ref{tab:dual_solution}, we specify a solution and then prove that it is indeed feasible for the dual LP. In this table, the constant $x$ is defined as $4410/338989303$.

\begin{table}\centering
\begin{tabular}{|c|c|c|c|c|c|c|c|c|c|c|}
\hline 
$i$ & 1 & 2&3&4&5&6&7&8&9&10 \\ 
\hline 
$\lambda_i/x$ &1&16&25&100&400&400&1600&6400&6400&25600 \\
\hline
$-\mu_i/x$& 863&3312&4125&8100&15600&14800&27200&44800&32000&25600 \\
\hline
\end{tabular}
\caption{The variable values for the dual solution. Here, $x=4410/338989303$.}
\label{tab:dual_solution}
\end{table}

Note that the dual constraint
\begin{equation*}
-\sum_{j=1}^{k}\mu_j \cdot \lim_{n \to \infty} \frac{OPT(L1\ldots L_j)}{n} \le 1
\end{equation*}
is satisfied with equality.

It thus remains to check the other constraints, where we have one constraint for every pattern. For verifying that all constraints 
$\sum_{i=j}^{k}\lambda_ip_i + \sum_{i=j}^{k}\mu_i \le 0$
are satisfied, we see that it suffices to check that for every $j=1,\ldots, k$ the inequality
$\max_{p \in T_j} \sum_{i=j}^{k}\lambda_ip_i \le -\sum_{i=j}^{k}\mu_i$ holds. We can interpret the $\lambda_i$ values as weights assigned to items of type $i$, and thus the problem reduces to finding the pattern in $T_j$ with maximum weight -- a knapsack problem. The $\mu_i$-values define the capacity of the knapsack.
In order to solve this efficiently, we introduce a dominance notion for items.

\begin{definition}
We say that $m^2$ items of size $s_i$ dominate an item of size $s_j$, denoted by $m^2 s_i \succ s_j$, if $m s_i \le s_j$ and $m^2 \lambda_i \ge \lambda_j$.
\end{definition}

In the case that $m^2 s_i$-items dominate an $s_j$-item, we can replace one item of size $s_j$ by $m^2$ items of size $s_i$ (arranged in an $m\times m$ grid), as the items to do not take more space. Furthermore, the weight of the pattern only increases by this replacement step, so it suffices to only examine the pattern without $s_j$ items.
Note that the $\succ$-operator is transitive.

For our input, we use the following dominance relations:
\begin{align*}
4^2 s_1 \succ s_2 && 5^2 s_1 \succ s_3 && 2^2 s_3 \succ s_4 \\ 2^2 s_4 \succ s_5 && s_5\succ s_6 && 2^2 s_6 \succ s_7\\
2^2 s_7 \succ s_8 && s_8 \succ s_9 && 2^2 s_9 \succ s_{10}
\end{align*}

It is easy to check that these are indeed fulfilled by the $\lambda_i$-values given above. The dominance relations give us that whenever a pattern contains $s_1$-items, we can replace all other items in this pattern by $s_1$-items as well -- thus, for set $T_1$, we only need to consider the pattern that contains only $s_1$-items (and the maximal number of them, i.e., $176400$ such items). So using the dominance relation, we have reduced the number of patterns dramatically. Similarly, for $T_3, \ldots, T_{10}$, we only have to consider one pattern each. Only for $T_2$, we have to be careful: As $s_3$-items are not dominated by $s_2$ items, we also have to consider patterns that contain $s_2$ and $s_3$-items. The following Lemma will show that the pattern that contains $6889$ $s_3$-items and $207$ $s_2$-items is the maximum weight pattern for this case.

\begin{lemma}
Among all patterns that only contain items of sizes $s_2$ and $s_3$, the pattern with $6889$ $s_3$-items and $207$ $s_2$-items has the highest weight given the $\lambda_i$-values of Table \ref{tab:dual_solution}.
\end{lemma}
\begin{proof}
Let $p^*$ be the pattern under consideration. We note that $\frac{\lambda_2}{(1/105)^2}<\frac{\lambda_3}{(1/84)^2}$, i.e., the weight per area is larger for the $s_3$-items than for the $s_2$-items. Furthermore, $p^*$ occupies an area of $\left(\frac{104}{105}\cdot (1+\epsilon)\right)^2$. Note that no pattern with these two item types can cover a larger area. Hence, no pattern can achieve a larger weight.
\end{proof}
In Table \ref{tab:patterns_for_dual}, we list all patterns that need to be checked, together with their weight and the knapsack capacity.

\begin{table}
\centering
\begin{tabular}{|c|c|c|c|}
\hline 
$j$ & heaviest pattern $p$ from $T_j$ & $w(p)/x$ & knapsack capacity: $(\sum_{i=j}^{10}-\mu_i)/x$ \\ 
\hline 
1 & $176400 \times s_1$ & $176400$ & $176400$ \\ 
\hline 
2 & $207 \times s_2, 6889 \times s_3$ & $175537$ & $175537$ \\ 
\hline 
3 & $6889 \times s_3$ & $172225$ & $172225$ \\ 
\hline 
4 & $1681 \times s_4$ & $168100$ & $168100$ \\ 
\hline 
5 & $400 \times s_5$ & $160000$ & $160000$ \\ 
\hline 
6 & $361 \times s_6$ & $144400$ & $144400$ \\ 
\hline 
7 & $81 \times s_7$ & $129600$ & $129600$ \\ 
\hline 
8 & $16 \times s_8$ & $102400$ & $102400$ \\ 
\hline 
9 & $9 \times s_9$ & $57600$ & $57600$ \\ 
\hline 
10 & $1\times s_{10}$ & $25600$ & $25600$ \\ 
\hline 
\end{tabular} 
\caption{The patterns that have to be considered to verify the first set of constraints in the dual LP. Again, $x=4410/338989303$.}
\label{tab:patterns_for_dual}
\end{table}

Finally, in order to determine the lower bound proven by this input, we compute $839\lambda_1 + 10\lambda_2 + 8\lambda_3 + 4\lambda_4 + 39\lambda_5+8\lambda_6+4\lambda_7+7\lambda_8+5\lambda_9+\lambda10 = 569767590/338989303 > 1.680783$. This concludes the proof of Theorem \ref{thm:general_lb}.

\section{Lower Bounds for General Algorithms for Rectangle Packing}

In this section, we present a lower bound on the more general two-dimensional bin packing, where items are allowed to be arbitrary rectangles and not necessarily squares. In this setting, we receive a sequence of $n$ items, where the $i$-th item has width $w_i$ and height $h_i$. The bins are still squares of side length one, and we are not allowed to rotate the items. Note that the LP and its dual are still the same, however, we need to adapt our definition of item dominance as follows.
\begin{definition}
We say that $m_1 \times m_2$ items of size $s_i$ dominate an item of size $s_j$, denoted by $(m_1 \times m_2) s_i \succ s_j$, if $m_1 w_i \le w_j, m_2 h_i \le h_j$ and $m_1m_2\lambda_i \ge \lambda_j$. Instead of $(1 \times 1) s_i \succ s_j$, we will simply write $s_i \succ s_j$.
\end{definition}
Again, this means that we can replace one item of size $s_j$ by $m_1m_2$ items of size $s_i$ which are arranged in an $m_1 \times m_2$-grid, while only increasing the weight (sum of $\lambda$-values) of the pattern.


\subsection{A lower bound of 1.859 using nine item types}

The construction for our lower bound relies on nine item types that are arranged in 3 groups, also called \textit{levels}. Corresponding to these item types, we have nine lists $L_1, \ldots, L_9$, where each $L_j$ consists of $n$ items of size $s_j$. The item sizes are given in Table \ref{tab:input_rect}.

\begin{table}
\centering
\begin{tabular}{|c|c|c|}
\hline
\multicolumn{3}{|c|}{Level 1}\\
\hline 
$j$ & $w_j$ & $h_j$ \\ 
\hline 
1 & $1/4-300\delta$ & \multirow{3}{*}{$1/6-2\epsilon$} \\ 
\cline{1-2}
2 & $1/4+100\delta$ &  \\ 
\cline{1-2}
3 & $1/2+200\delta$ &  \\ 
\hline 
\end{tabular} 
\begin{tabular}{|c|c|c|}
\hline
\multicolumn{3}{|c|}{Level 2}\\
\hline 
$j$ & $w_j$ & $h_j$ \\ 
\hline 
4 & $1/4-30\delta$ & \multirow{3}{*}{$1/3+\epsilon$} \\ 
\cline{1-2}
5 & $1/4+10\delta$ &  \\ 
\cline{1-2}
6 & $1/2+20\delta$ &  \\ 
\hline 
\end{tabular} 
\begin{tabular}{|c|c|c|}
\hline
\multicolumn{3}{|c|}{Level 3}\\
\hline 
$j$ & $w_j$ & $h_j$ \\ 
\hline 
1 & $1/4-3\delta$ & \multirow{3}{*}{$1/2+\epsilon$} \\ 
\cline{1-2}
2 & $1/4+\delta$ &  \\ 
\cline{1-2}
3 & $1/2+2\delta$ &  \\ 
\hline 
\end{tabular} 
\caption{The input sequence for the 1.859 lower bound. The tables show the items for the first, second and third level. $\epsilon$ and $\delta$ are assumed to be sufficiently small positive constants.}
\label{tab:input_rect}
\end{table}

First of all, we will again give the values for $OPT(L_1\ldots L_j)$ in Table \ref{tab:opt_rect}. The packings for the optimal solution are depicted in Figure \ref{fig:opt_rect}. 

\begin{table}
\centering
\begin{tabular}{|c|c|c|c|c|c|c|c|c|c|}
\hline 
$j$ & 1 & 2 & 3 & 4 & 5 & 6 & 7 & 8 & 9 \\ 
\hline 
$OPT(L_1\ldots L_j)/n$ & $1/24$ & $1/12$ & $1/6$ & $1/4$ & $1/3$ & $1/2$ & $5/8$ & $3/4$ & $1$ \\ 
\hline 
\end{tabular} 
\caption{The optimal solution values for the 1.859 lower bound.}
\label{tab:opt_rect}
\end{table}


\begin{figure}
\centering
\subfloat[$OPT(L_1)$]{\includegraphics[width=0.2\textwidth]{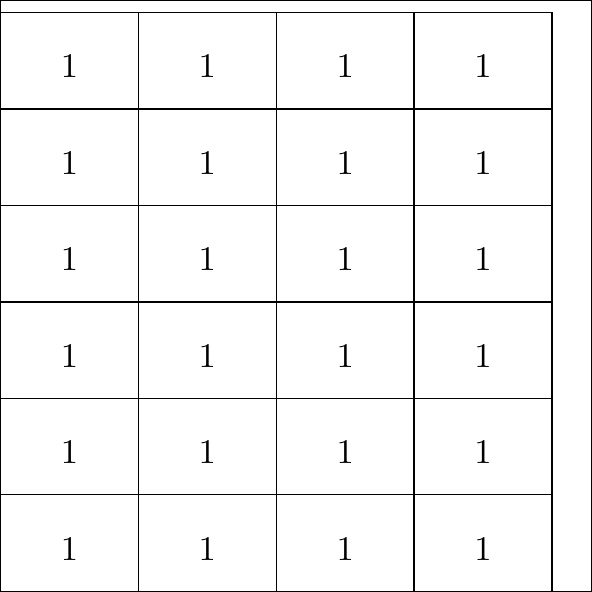}
}
\quad
\subfloat[$OPT(L_1L_2)$]{\includegraphics[width=0.2\textwidth]{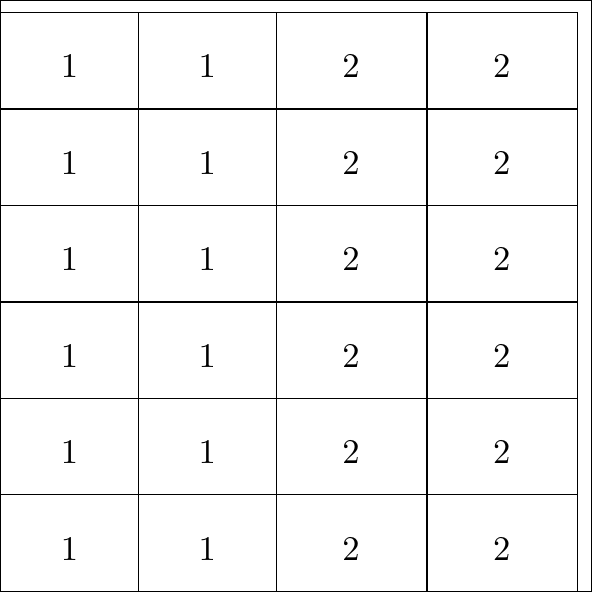}
}
\quad 
\subfloat[$OPT(L_1L_2L_3)$]{\includegraphics[width=0.2\textwidth]{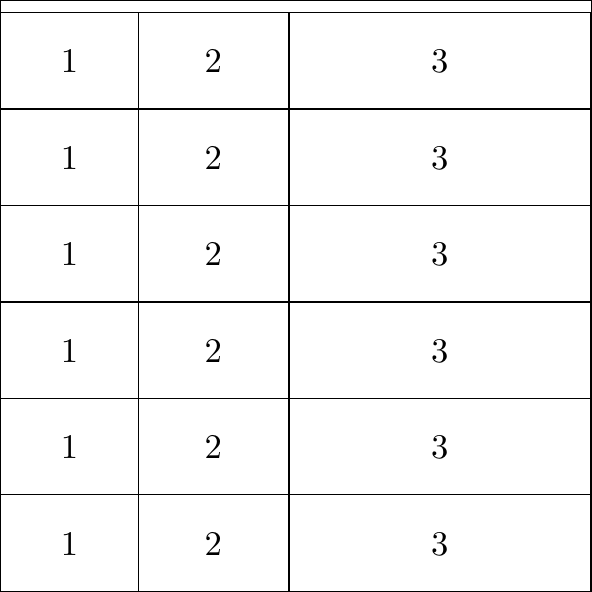}}

\subfloat[$OPT(L_1\ldots L_4)$]{\includegraphics[width=0.2\textwidth]{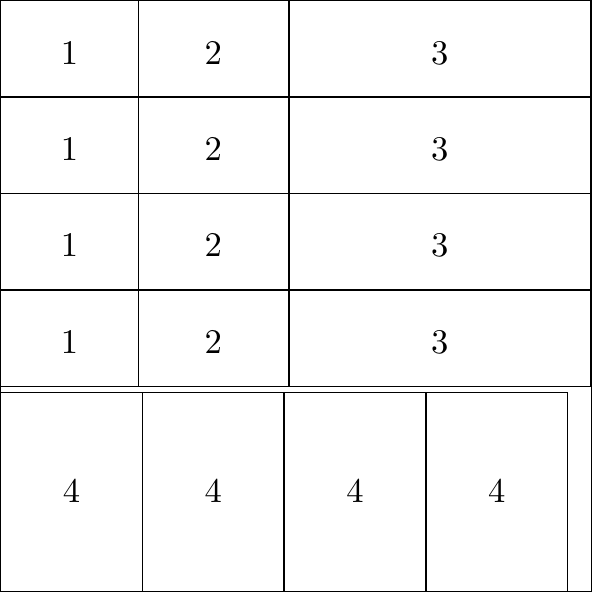}
}
\quad
\subfloat[$OPT(L_1\ldots L_5)$: We have $n/4$ bins with the left packing and $n/12$ bins with the right packing]{\includegraphics[width=0.2\textwidth]{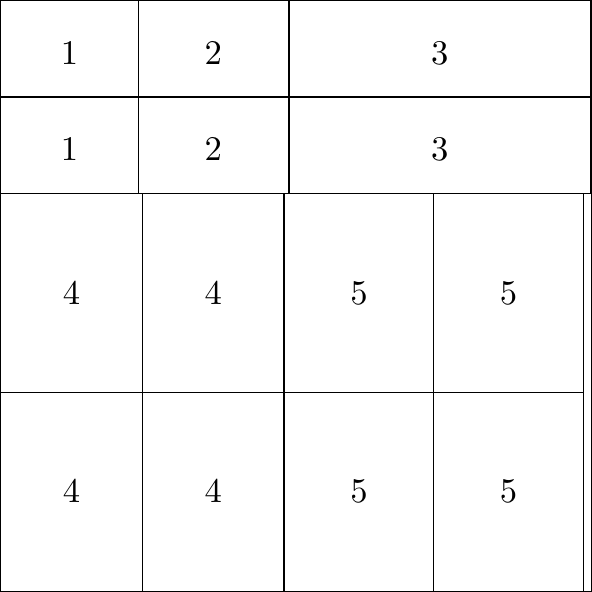}$\,\,$\includegraphics[width=0.2\textwidth]{opt_rect_3_sixth}
}

\subfloat[$OPT(L_1\ldots L_6)$]{\includegraphics[width=0.2\textwidth]{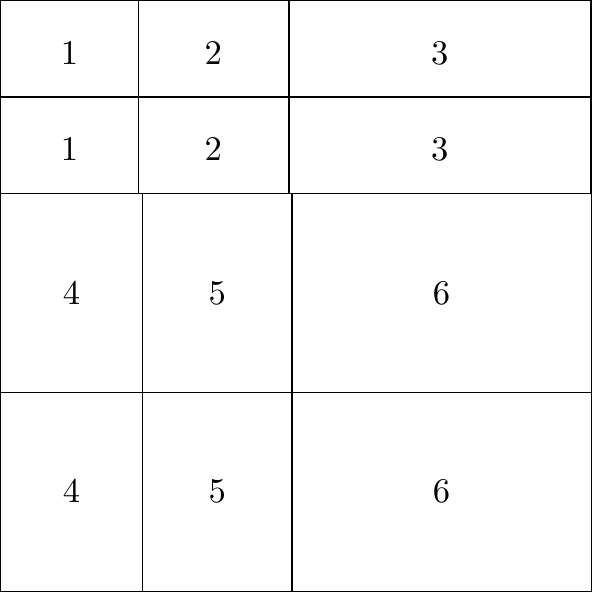}}
\quad
\subfloat[$OPT(L_1\ldots L_7)$: We have $n/4$ bins with the left packing and $3/8\cdot n$ bins with the right packing]{\includegraphics[width=0.2\textwidth]{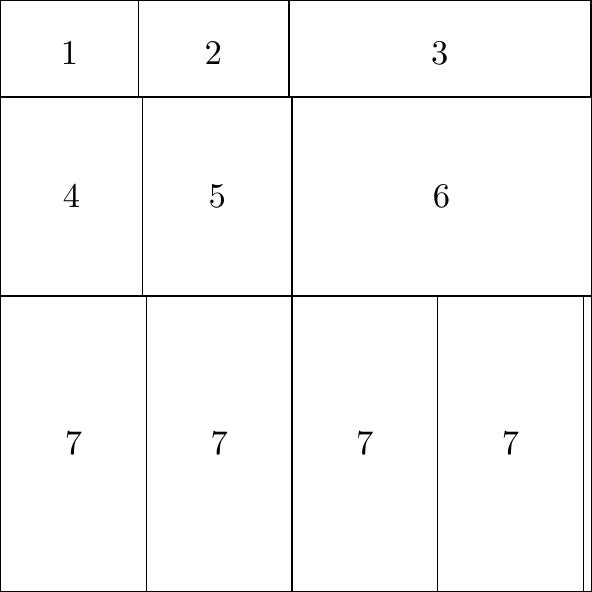}$\,\,$ \includegraphics[width=0.2\textwidth]{opt_rect_6_sixth}
}

\subfloat[$OPT(L_1\ldots L_8)$: We have $n/2$ bins with the left packing and $n/4$ bins with the right packing]{\includegraphics[width=0.2\textwidth]{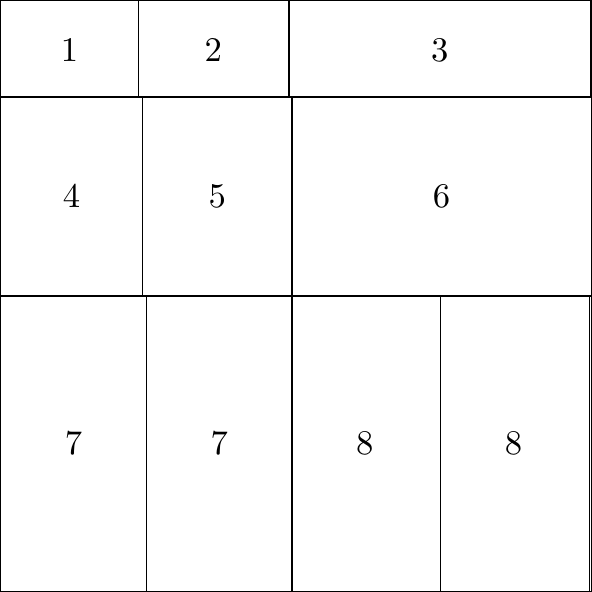}$\,\,$ \includegraphics[width=0.2\textwidth]{opt_rect_6_sixth}
}
\quad 
\subfloat[$OPT(L_1\ldots L_9)$]{\includegraphics[width=0.2\textwidth]{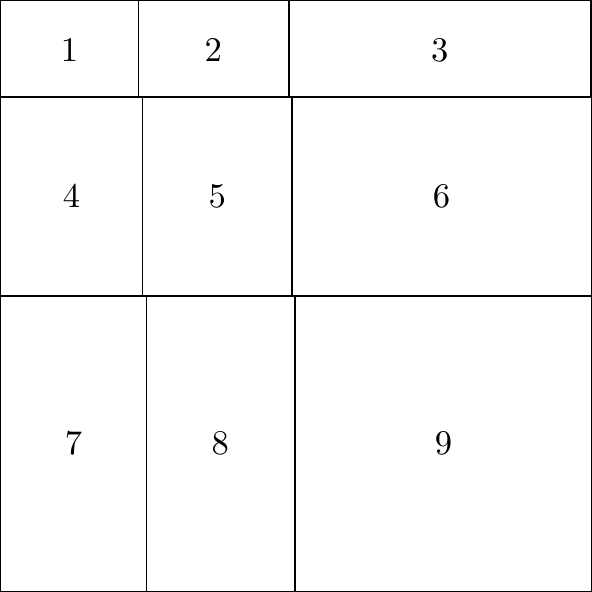}}
\caption{Optimal solutions for sublists $L_1\ldots L_j$ for all $j=1,\ldots,9$. The number within every item denotes its type.}
\label{fig:opt_rect}
\end{figure}

We will now give an optimal solution for the primal LP and then verify its feasibility by checking the constraints, and verify its optimality by giving a matching dual solution and a proof of its feasibility. The optimal primal solution uses 15 different patterns as listed in Table \ref{tab:patterns_rect}. The packings for some of these (where it is not that easy to see how to pack the pattern) are depicted in Figure \ref{fig:patterns_rect}.

\begin{table} \centering
\begin{tabular}{|c|ccccccccc|c|c|}
\hline 
pattern $p$ & $p_1$ & $p_2$ & $p_3$ & $p_4$ & $p_5$ & $p_6$ & $p_7$ & $p_8$ & $p_9$ & type & $x(p)$ \\ 
\hline 
$p^{(1)}$ & 24 & - & - & - & - & - & - & - & - & \multirow{3}{*}{$T_1$} & $29/4956$ \\ 
\cline{1-10}\cline{12-12}
$p^{(2)}$ & 12 & 12 & - & - & - & - & - & - & - &  & $289/4956$ \\ 
\cline{1-10}\cline{12-12}
$p^{(3)}$ & 12 & - & 6 & - & - & - & - & - & - &  & $11/826$ \\ 
\hline 
$p^{(4)}$ & - & 6 & 6 & - & - & - & - & - & - & \multirow{3}{*}{$T_2$} & $15/413$ \\ 
\cline{1-10}\cline{12-12}
$p^{(5)}$ & - & 2 & 2 & 4 & 4 & - & - & - & - &  & $17/1239$ \\ 
\cline{1-10}\cline{12-12}
$p^{(6)}$ & - & 2 & 2 & 4 & - & 2 & - & - & - &  & $34/1239$ \\ 
\hline 
$p^{(7)}$ & - & - & 4 & 2 & 4 & - & - & - & - & $T_3$ & $64/413$ \\ 
\hline 
$p^{(8)}$ & - & - & - & 8 & - & - & - & - & - & \multirow{3}{*}{$T_4$} & $55/1239$ \\ 
\cline{1-10}\cline{12-12}
$p^{(9)}$ & - & - & - & 2 & 2 & 2 & - & - & - &  & $74/1239$ \\ 
\cline{1-10}\cline{12-12}
$p^{(10)}$ & - & - & - & 1 & 1 & 1 & 4 & - & - &  & $21/413$ \\ 
\hline 
$p^{(11)}$ & - & - & - & - & 1 & 1 & 2 & 2 & - & \multirow{3}{*}{$T_5$} & $32/413$ \\ 
\cline{1-10}\cline{12-12}
$p^{(12)}$ & - & - & - & - & 1 & 1 & 4 & - & - &  & $3/413$ \\ 
\cline{1-10}\cline{12-12} 
$p^{(13)}$ & - & - & - & - & 1 & 1 & 1 & 1 & 1 &  & $29/413$ \\ 
\hline 
$p^{(14)}$ & - & - & - & - & - & 2 & 1 & 1 & - & $T_6$ & $128/413$ \\ 
\hline 
$p^{(15)}$ & - & - & - & - & - & - & 1 & 1 & 1 & $T_7$ & $96/413$ \\ 
\hline 
$p^{(16)}$ & - & - & - & - & - & - & - & 1 & 1 & $T_8$ & $96/413$ \\ 
\hline 
$p^{(17)}$ & - & - & - & - & - & - & - & - & 1 & $T_9$ & $192/413$ \\ 
\hline 
\end{tabular} 
\caption{The optimal primal solution for the 1.859 lower bound. The table gives the patterns, the set $T_j$ they belong to, and the value of the LP variable $x(p)$ for each pattern.}
\label{tab:patterns_rect}
\end{table}

\begin{figure}
\subfloat[$p^{(6)}$]{\includegraphics[width=0.3\textwidth]{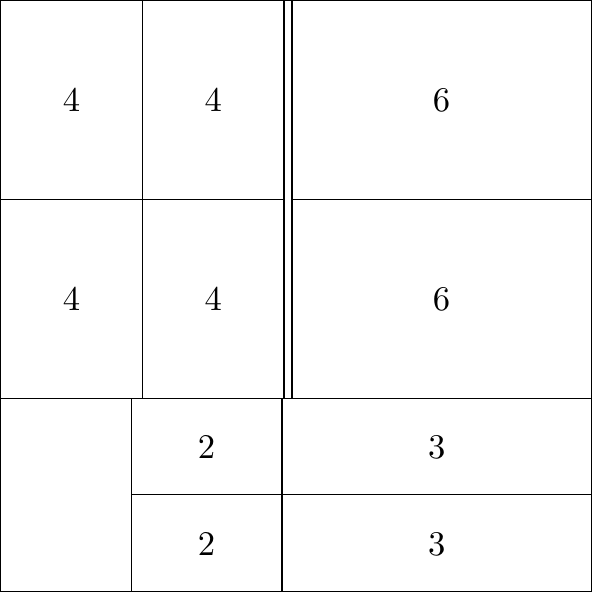}}
\quad 
\subfloat[$p^{(7)}$]{\includegraphics[width=0.3\textwidth]{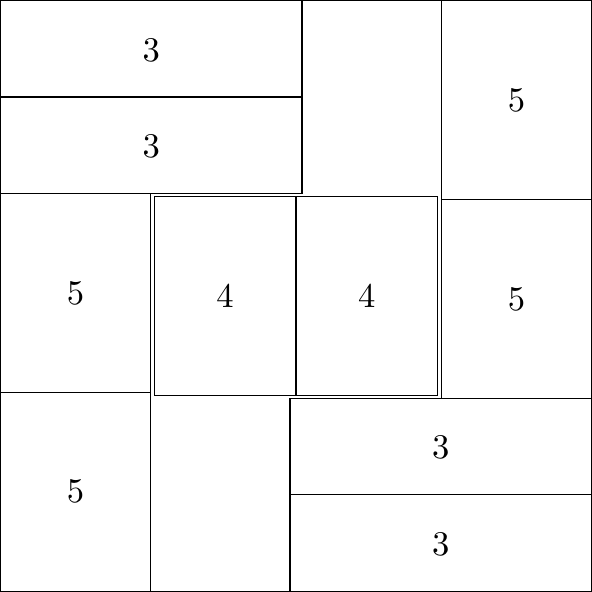}}
\quad 
\subfloat[$p^{(14)}$]{\includegraphics[width=0.3\textwidth]{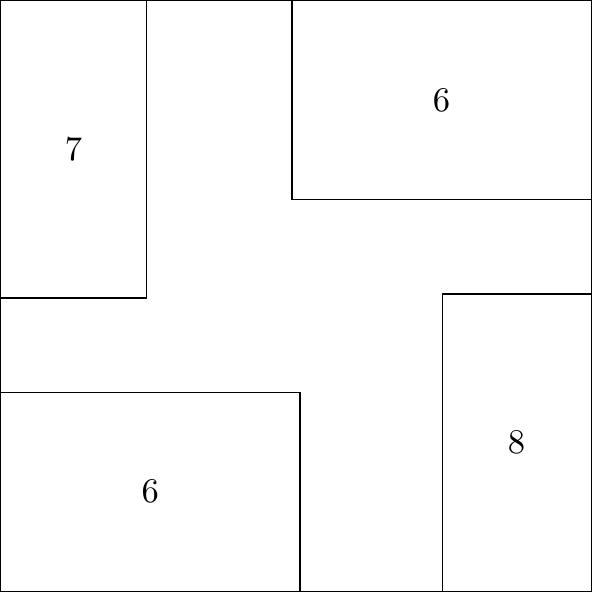}\label{subfig:patterns_rect_14}}
\caption{Packings of some of the patterns of the optimal LP solution. The number within each item denotes its type.}
\label{fig:patterns_rect}
\end{figure}

To verify feasibility of this solution, note that all the constraints hold with equality (except for the non-negativity constraints of course) if we set $R=768/413 > 1.859$. For proving optimality, we use the dual solution given in Table \ref{tab:dual_solution_rect}. In order to verify its feasibility, note that the dual LP constraint 
\begin{equation*}
-\sum_{j=1}^{k}\mu_j \cdot \lim_{n \to \infty} \frac{OPT(L1\ldots L_j)}{n} \le 1
\end{equation*}
is satisfied with equality. It remains to check the other dual constraints, where again we have to test whether $\max_{p \in T_j} \sum_{i=j}^k\lambda_ip_i \le -\sum_{i=j}^k\mu_i$ for all $j=1,\ldots ,k$. For solving the associated knapsack problem, we use the following dominance relations to simplify our task:
\begin{center}
$\begin{array}{lll}
s_1 \succ s_2 \quad & (2\times 1) s_2 \succ s_3 \quad & (1\times 2) s_1 \succ s_4 \\ 
s_4 \succ s_5 \quad & (2\times 1) s_5\succ s_6 \quad & s_4 \succ s_7\\
s_7 \succ s_8 \quad & (2\times 1) s_8 \succ s_9 \quad&
\end{array}$
\end{center}
We list the optimal patterns to be considered for the knapsack problem in Table \ref{tab:patterns_for_dual_rect}. In the following lemmas, we will prove that it suffices to consider these patterns.

\begin{table}
\centering
\begin{tabular}{|c|c|c|c|c|}
\hline 
$j$ & types to consider & heaviest patterns $p$ from $T_j$ & $w(p)\cdot 413$ & $(\sum_{i=j}^9 -\mu_i)\cdot 413$ \\ 
\hline 
1 & 1 & $24\times s_1$ & 1152 & 1152 \\ 
\hline 
\multirow{3}{*}{2} & \multirow{3}{*}{2, 4} & $18 \times s_2$ & \multirow{3}{*}{864} & \multirow{3}{*}{864} \\ 
\cline{3-3}
&& $6\times s_2, 8\times s_4$ &  & \\
\cline{3-3}
&& $12\times s_2, 4\times s_4$ &  & \\
\hline
3 & 3, 4 & $4\times s_3, 6\times s_4$ & 816 & 816 \\ 
\hline 
4 & 4 & $8\times s_4$ & 576 & 576 \\ 
\hline 
5 & 5,7& $3\times s_5, 4\times s_7$ & 504 & 504 \\ 
\hline 
\multirow{2}{*}{6} & \multirow{2}{*}{6,7} & $2\times s_6, 2\times s_7$ & \multirow{2}{*}{432} & \multirow{2}{*}{432} \\ 
\cline{3-3} 
&& $1\times s_6, 4\times s_7$ & & \\
\hline
7 & 7& $4\times s_7$ & 288 & 288 \\ 
\hline 
8 &8& $3\times s_8$ & 216 & 216 \\ 
\hline 
9 &9& $1\times s_9$ & 144 & 144 \\ 
\hline 
\end{tabular} 
\caption{The patterns that need to be considered for verifying the feasibility of the dual solution, together with their weight and the knapsack capacity.}
\label{tab:patterns_for_dual_rect}
\end{table}

\begin{lemma}
For $T_2$, the patterns $p^{(1)}=(0,6,0,8,0,\ldots,0), p^{(2)}=(0,12,0,4,0,\ldots,0),$ and $p^{(3)}=(0,18,0,\ldots, 0)$ maximize $w(p) = \sum_{i=2}^k\lambda_i p_i$, given the $\lambda_i$-values from Table \ref{tab:dual_solution_rect}.
\end{lemma}
\begin{proof}
%
In this proof, we abbreviate patterns by listing only their second and fourth components.
Any vertical line through a bin can intersect with at most two type 4 items, and any horizontal line with at most four. By arranging the items in two rows of four, we see that $(0,8)$ is a dominent pattern.
There are only two options for a horizontal line in a bin that contains only
type 2 and type 4 items:
\begin{itemize}
\item it crosses (at most) four type 4 items
\item it crosses at most three items (of any type).
\end{itemize}
It is easy to see that $(0,8)$ is a (dominant) pattern: we can have
two rows of four items.
In general, if a bin contains eight type 4 items, there is at least a height of
$1/3-2\varepsilon$ where horizontal lines cross with at most three items
(this height can be more if the type 4 items are not exactly aligned).
This in turn implies that a volume of at least $(1/3-2\varepsilon)(1/4-300\delta)$
must remain empty in any bin that contains eight type 4 items, as the maximum total width of a set of three items is $3/4+300\delta$. It follows immediately
that $(6,8)$ is a dominant pattern, as the free space in the packing tends to exactly $1/3\times 1/4$ if $\varepsilon \rightarrow 0$ and $\delta \rightarrow 0$ (and since it is indeed a pattern).

If there is a total height of more than $1/3+\varepsilon$ at which a horizontal line intersects with four items, then by considering the highest and the lowest such line, we can identify eight distinct type 4 items.
Therefore, in a bin with four to seven type 4 items, at a height of at least $2/3-\varepsilon$, a horizontal line intersects with at most three items, since you can only have one row of four type 4 items. Therefore, in such a bin, there must be $(2/3-\varepsilon)(1/4-300\delta)$ of empty space. We find the following patterns:
$(6,7),(8,6),(10,5),(12,4)$ with weights $\frac{792}{413}$, $\frac{816}{413}$, $\frac{840}{413}$, $\frac{864}{413}$, respectively.

If there are at most three type 4 items, then \emph{any} horizontal line intersects with at most three items, the empty space is at least $1/4-300\delta$, and the patterns are
$(12,3),(14,2),(16,1),(18,0)$ with weights $\frac{792}{413}$, $\frac{816}{413}$, $\frac{840}{413}$, $\frac{864}{413}$, respectively.
\end{proof}

\begin{lemma}\label{lem:packing_s3_s4_rect}
For $T_3$, the pattern $p=(0,0,4,6,0,\ldots, 0)$ maximizes $w(p) = \sum_{i=2}^k\lambda_i p_i$, given the $\lambda_i$-values from Table \ref{tab:dual_solution_rect}.
\end{lemma}
\begin{proof}
In this proof, we again abbreviate patterns by listing only their third and fourth components.
First of all, see Figure \ref{fig:packing_s3_s4_rect} for a feasible packing of $p$. 
%
Horizontal lines can only intersect with these sets of items:
\begin{itemize}
\item three or four type 4 items
\item at most \emph{two} items which have total width at most $3/4+170\delta$
\end{itemize}
As above, $(0,8)$ is a dominant pattern. There is at least a height of $1/3-2\varepsilon$ at which horizontal lines
intersect with at most two items. There is at most a height of $1/3-2\varepsilon$
at which horizontal lines intersect with at most one type 4 item.

Since two type 3 items cannot be placed next to
each other, 
this implies that $(2,8)$ is a (dominant) pattern, but with a smaller weight of $\frac{768}{413}$.

If there are seven type 4 items, again there is at most a height of $1/3-2\varepsilon$
at which horizontal lines intersect with at most one type 4 item, so $(3,7)$ is not a pattern.
If there are four to six type 4 items, there is an empty volume of at least
$(2/3-\varepsilon)(1/4-170\delta)$, so $(4,6)$ is dominant. Moreover, there is
at least a height of $2/3-\varepsilon$ at which horizontal lines
intersect with at most two items, so $(5,5)$ and $(5,4)$ are not patterns.

If there are three type 4 items, the empty volume is at least
$1/4-170\delta$, so $(6,3)$ is a dominant pattern with weight $\frac{792}{413}$. Finally, no bin can contain more than six type 3 items, so no other pattern can be heavier.
\end{proof}

\begin{figure}
\centering
\includegraphics[width=.25 \textwidth]{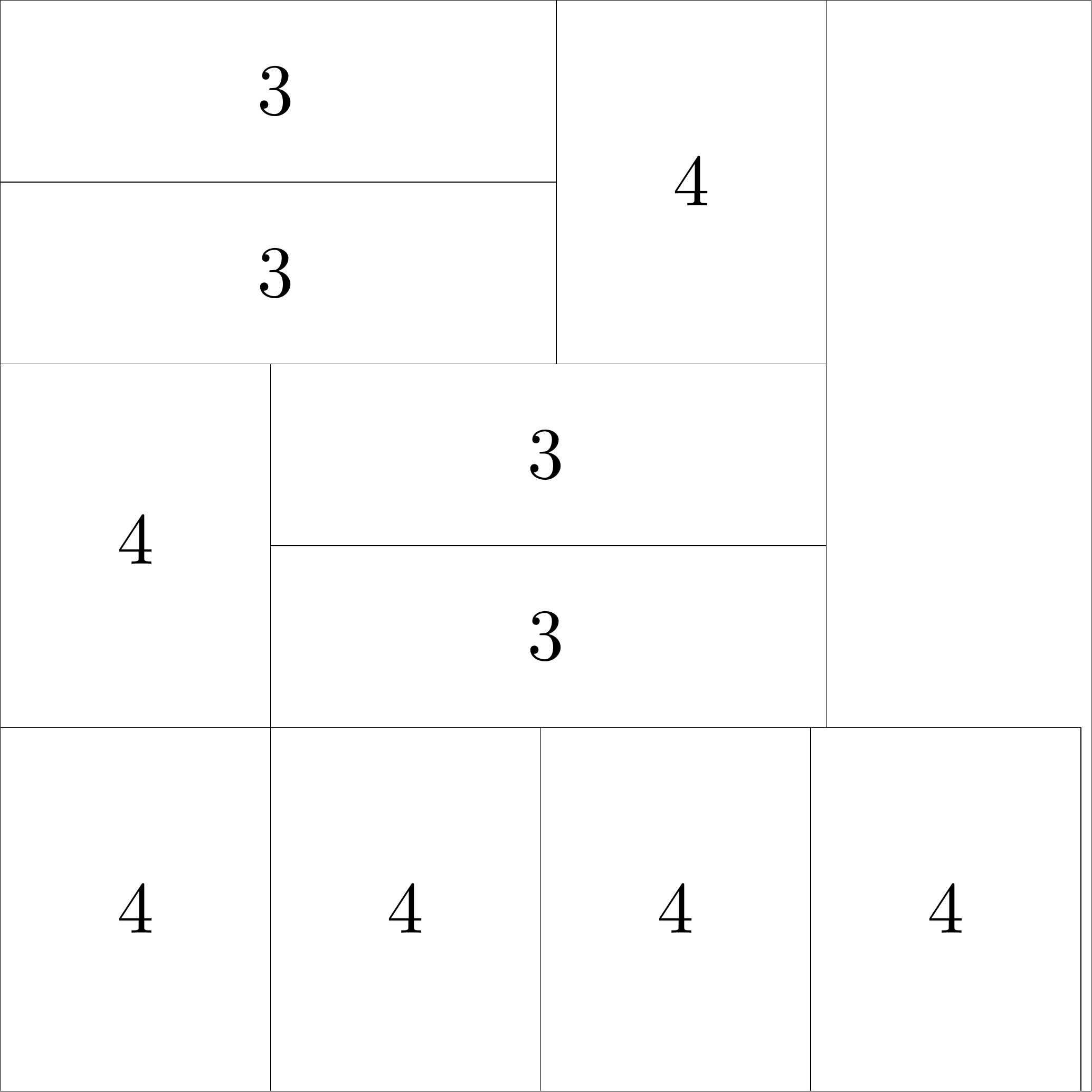}
\caption{A feasible packing for pattern $p$ from Lemma \ref{lem:packing_s3_s4_rect}}
\label{fig:packing_s3_s4_rect}
\end{figure}

\begin{lemma}
For $T_5$, the pattern $p=(0,0,0,0,3,0,4,0,0)$ maximizes $w(p) = \sum_{i=2}^k\lambda_i p_i$, given the $\lambda_i$-values from Table \ref{tab:dual_solution_rect}.
\end{lemma}
\begin{proof}
Note that $\lambda_5=\lambda_7$. Therefore the only question here is how many items
of these types can be packed together in a bin.
No more than four type 7 items can be packed in any bin, and at most three type 5
items can be packed with them (using similar arguments as above). Moreover, no more than six type 5 items can be packed in any bin, and if there are less than four type 7 items in any bin, then any horizontal line in such a bin intersects with at most three items. Thus at most seven items can be packed into any such bin.
\end{proof}

\begin{lemma}
For $T_6$, the patterns $p^{(1)}=(0,0,0,0,0,2,2,0,0)$ and $p^{(2)}=(0,0,0,0,0,1,4,0,0)$ maximize $w(p) = \sum_{i=2}^k\lambda_i p_i$, given the $\lambda_i$-values from Table \ref{tab:dual_solution_rect}.
\end{lemma}
\begin{proof}
For a packing of $p^{(1)}$, see Figure \ref{subfig:patterns_rect_14}. Observe that a bin can never contain more than two $s_6$-items, and together with the fact that it is easy to see that no more than two $s_7$-items can be added to them, it follows that $p^{(1)}$ is a candidate for the heaviest pattern. Likewise, it is clear that no bin can contain more than four $s_7$-items, and it is easy to see that no more than one $s_6$-item can be added to those.
\end{proof}

\begin{table}
\centering
\begin{tabular}{|c|c|c|c|c|c|c|c|c|c|}
\hline 
$j$ & 1 & 2 & 3 & 4 & 5 & 6 & 7 & 8 & 9 \\ 
\hline 
$\lambda_j \cdot 413$ & 48 & 48 & 96 & 72 & 72 & 144 & 72 & 72 & 144 \\ 
\hline 
$-\mu_j\cdot 413$ & 288 & 48 & 240 & 72 & 72 & 144 & 72 & 72 & 144 \\ 
\hline 
\end{tabular} 
\caption{Dual solution for the 1.859 lower bound.}
\label{tab:dual_solution_rect}
\end{table}

\section{Lower Bound for Harmonic-Type Algorithms}
\label{sect:harmonic}

Now, we consider the hypercube packing problem in $d$ dimensions, for any $d \ge 2$.
We define the class \class{} of Harmonic-type algorithms analogous to \cite{RBLL89}. 
An algorithm $\mathcal{A}$ in \class{} for any $h \ge 1$ distinguishes, possibly among others, the following disjoint subintervals
\begin{itemize}
\item $\overline{I}_1 = \left( 1-y_1,1 \right]$
\item $I_{1,j} = \left( 1-y_{j+1}, 1-y_j \right]$, for every $j \in \{1, \ldots, h\}$
\item $\overline{I}_2 = \left( y_h, 1/2 \right]$
\item $I_{2,j} = \left( y_{h-j}, y_{h-j+1} \right]$, for every $j \in \{1, \ldots, h\}$
\item $I_\lambda = \left( 0, \lambda \right]$
\end{itemize}
for some parameters $y_j$ and $\lambda$, where $1/3 = y_0 < y_1 < \ldots < y_h < y_{h+1} = 1/2$ and $0 < \lambda \le 1/3$. 
For convenience, we assume that all $y_j$ are rational.

Algorithm $\mathcal{A}$ has to follow the following rules:
\begin{enumerate}
\item For each $j \in \{1, \ldots, h\}$, there is a constant $m_j$ s.t. a $1/m_j$-fraction of the items of side length in $I_{2,j}$ is packed $2^d-1$ per bin (``red items''), the rest are packed $2^d$ per bin (``blue items'').
\item No bin contains an item of side length in $I_{1,i}$ and an item of side length in $I_{2,j}$ if $i+j \le h$.
\item No bin contains an item of side length in $\overline{I}_1$ \textbf{and} an item of side length in $I_{2,j}$.
\item No bin contains an item of side length in $I_{1,j}$ \textbf{and} an item of side length in $\overline{I}_2$.
\item No bin that contains an item of side length in $I_\lambda$ contains an item of side length in $I_{1,j}, I_{2,j}, \overline{I}_1$ or $\overline{I}_2$.
\end{enumerate}

We will now define $2h+1$ input instances for the hypercube packing problem in $d$ dimensions, and for each instance we derive a lower bound on the number of bins any \class{}-algorithm must use to pack this input.

Every such input instance consists of three types of items. 
The input will contain $N$ items of side length $u$, followed by $(2^d-1)N$ items of side length $v$ and finally followed by $MN$ items of side length $t$, where $u,v,t$ and $M$ will be defined for each instance differently. 
We will then show, for every instance, that one $u$-item, $2^d-1$ $v$-items and $M$ $t$-items can be packed together in one bin, thus the optimal packing for this input uses at most $N$ bins.

\subsection{Instances $1, \ldots, h$}

Let $\epsilon>0$ be arbitrarily small. For every $j \in \{1, \ldots, h\}$, we define the following instance of the problem: 
Let $u=\frac{1+\epsilon}{2}, v=(1+\epsilon)y_{h-j}$ and $t=\frac{(1+\epsilon)y_{h-j}}{2K}$ for some large integer $K$ such that $t \in I_\lambda$ and $\frac{K}{y_{h-j}} \in \mathbb{N}$. 
Clearly, $u \in I_{1,h}$ and $v \in I_{2,j}$.

In order to show that one $u$-item, $2^d-1$ $v$-items and $M$ $t$-items can be packed in one bin, we will define anchor points for each size and then place items at some of these such that no two items are overlapping.

There is only one anchor point for $u$-items, namely $(0, \ldots, 0)$, i.e. the origin of the bin. 
We place one $u$ item there.
For items of side length $v$, we define anchor points as all points having all coordinates equal to $(1+\epsilon)/2$ or $(1+\epsilon)/2-(1+\epsilon)y_{h-j}$. 
This defines $2^d$ anchor points, but an anchor point can only be used for a $v$-item if at least one coordinate is $(1+\epsilon)/2$. 
Hence, we can pack $2^d-1$ $v$-items together with the $u$-item placed before.

For items of side length $t$, the anchor points are all points with coordinates equal to $i\frac{(1+\epsilon)y_{h-j}}{2K}$ for $i=0, \ldots, \frac{2K}{y_{h-j}}-2$, i.e. we have $(\frac{2K}{y_{h-j}}-1)^d$ anchor points for these items. 
These anchor points form a superset of all previous anchor points for $u$- and $v$-items. 
Together with the fact that $t$ divides $u$ and $v$, we can conclude that all larger items take away an integer amount of anchor points for the $t$-items. 
To be precise, the $u$-item blocks $\left(u/t\right)^d=\left(K/y_{h-j}\right)^d$ anchor points for $t$-items and each $v$-item blocks $\left(\frac{v}{t}\right)^d=(2K)^d$ anchor points for $t$-items. 
Hence, we can add $M := \left(\frac{2K-y_{h-j}}{y_{h-j}}\right)^d - \left(\frac{K}{y_{h-j}}\right)^d-(2^d-1)(2K)^d$ $t$-items to the items packed before.

A Harmonic-type algorithm $\mathcal{A}$ packs a $1/m_j$-fraction of the $N(2^d-1)$ $v$-items $2^d-1$ per bin, using $\frac{(2^d-1)N/m_j }{2^d-1} = \frac{N}{m_j}$ bins in total. 
The remaining $N(2^d-1)(1-1/m_j)$ $v$-items are packed $2^d$ per bin, adding another $ N(1-1/m_j)\frac{2^d-1}{2^d}= N(1-1/m_j)\left(1-\frac{1}{2^d}\right)$ bins.

$N/m_j$ of the $u$-items are added to bins with red $v$-items, the remaining $N(1-1/m_j)$ items of side length $u$ must be packed one per bin.

Finally, an algorithm in the class \class{} needs at least $NM/\left( \frac{2K-y_{h-j}}{y_{h-j}} \right)^d$ bins to pack the $t$-items, giving $$N\left(1 - \left(\frac{K}{2K-y_{h-j}}\right)^d-(2^d-1)\left(\frac{2Ky_{h-j}}{2K-y_{h-j}}\right)^d\right)$$ bins for these items. 
If we let $K \rightarrow \infty$, this tends to $N\left( 1-1/2^d-(2^d-1)y_{h-j}^d \right)$.

So, the total number of bins needed is at least 
\begin{align*}
& N\left( \frac{1}{m_j} + \left(1-\frac{1}{m_j}\right)\left(1-\frac{1}{2^d}\right) + 1-\frac{1}{m_j} + 1 - \frac{1}{2^d}-(2^d-1)y_{h-j}^d \right)\\
& = N\left( 2 + \left(1-\frac{1}{m_j}\right)\left(1-\frac{1}{2^d}\right) - \frac{1}{2^d}-(2^d-1)y_{h-j}^d \right)
\end{align*}

As the optimal solution uses at most $N$ bins, the performance ratio of any such algorithm $\mathcal{A}$ must be at least 
\begin{equation}
R_\mathcal{A} \ge 2 + (1-1/m_j)(1-1/2^d) - 1/2^d-(2^d-1) y_{h-j}^d \quad\quad\quad j=1,\ldots,h \label{eq:ineq1}
\end{equation}

\subsection{Instances $h+1, \ldots, 2h$}
\label{subsec:harmonic2}

Another set of instances is given for any $j\in \{1, \ldots, h\}$, if we use $u=(1+\epsilon)(1-y_{h-j+1}), v=(1+\epsilon)y_{h-j}$ and $t=\frac{(1+\epsilon)y_{h-j}(1-y_{h-j+1})}{K}$ for some large enough integer $K$ such that $u \in I_{1,h-j}, v \in I_{2,j}, t \in I_\lambda$ and $\frac{K}{y_{h-j}}, \frac{K}{1-y_{h-j+1}} \in \mathbb{N}$. 
For these item sizes, the algorithm is not allowed to combine $u$-items with $v$-items in the same bin, although space for items in $I_{1,i}$ with $i>h-j$ is reserved in red bins containing $v$-items.
We define the following anchor points: the point $(0,0)$ for type $u$; all points with all coordinates equal to $(1+\epsilon)(1-y_{h-j+1})$ or $(1+\epsilon)(1-y_{h-j+1})-(1+\epsilon)y_{h-j}$ for type $v$; and all points with all coordinates equal to $i\frac{(1+\epsilon)y_{h-j}(1-y_{h-j+1})}{K}$ for some $i \in \{0, \ldots, \frac{K}{y_{h-j}(1-y_{h-j+1})}-2\}$ for type $t$.
Again the anchor points for $u$- and $v$-items are a subset of the anchor points for $t$-items, and hence with the same argumentation as before we can pack one $u$-item together with $2^d-1$ $v$-items and $M$ $t$-items if we choose $M=\left( \frac{K-y_{h-j}(1-y_{h-j+1})}{y_{h-j}(1-y_{h-j+1})} \right)^d - \left(\frac{K}{y_{h-j}}\right)^d - (2^d-1)\left( \frac{K}{1-y_{h-j+1}} \right)^d$, as the $u$-item takes up $\left(\frac{K}{y_{h-j}}\right)^d$ anchor points of the $t$-items and each $v$-item takes up $\left( \frac{K}{1-y_{h-j+1}} \right)^d$ of these anchor points.

A similar calculation to before can be done:
An algorithm in class \class{} needs $N/m_j + N(1-1/m_j)(1-1/2^d)$ bins for red and blue items of type $v$. 
It needs $N$ bins for $u$-items, as they are packed one per bin, and finally
\begin{align*} & \frac{NM}{\left(\frac{K-y_{h-j}(1-y_{h-j+1})}{y_{h-j}(1-y_{h-j+1})}\right)^d}\\
& = N \left( 1 - \left( \frac{K(1-y_{h-j+1})}{K-y_{h-j}(1-y_{h-j+1})} \right)^d - (2^d-1)\left( \frac{Ky_{h-j}}{K-y_{h-j}(1-y_{h-j+1})} \right)^d \right)\\
& \xrightarrow{K\rightarrow \infty} N\left( 1 - \left(1-y_{h-j+1}\right)^d - (2^d-1)y_{h-j}^d \right)
\end{align*}
bins are required to pack the $t$-items. 
Hence, we need at least
\begin{align*}
& N\left( \frac{1}{m_j} + \left(1-\frac{1}{m_j}\right)\left(1-\frac{1}{2^d}\right) + 1 + 1 - \left(1-y_{h-j+1}\right)^d - (2^d-1)y_{h-j}^d \right)\\
& = N\left( 2 + \frac{1}{m_j} + \left(1-\frac{1}{m_j}\right)\left(1-\frac{1}{2^d}\right) - \left(1-y_{h-j+1}\right)^d - (2^d-1)y_{h-j}^d \right)
\end{align*}
bins in total. This gives the following lower bound for the performance ratio:
\begin{eqnarray}
R_\mathcal{A} \ge 2 + 1/m_j + \left(1-1/m_j\right)\left(1-1/2^d\right) - \left(1-y_{h-j+1}\right)^d - (2^d-1)y_{h-j}^d \label{eq:ineq2}\\
\hfill j=1,\ldots,h \nonumber
\end{eqnarray}

\subsection{Instance $2h+1$}
\label{subsec:harmonic3}

Let $u=\frac{1+\epsilon}{2}, v=(1+\epsilon)y_h$ and $t=\frac{(1+\epsilon)y_h}{2K}$ for some large enough integer $K$ such that $u \in I_{1,h}, v \in \overline{I}_{2}, t \in I_\lambda$ and $\frac{K}{y_h} \in \mathbb{N}$. 
For these item sizes, the algorithm is not allowed to combine $u$-items with $v$-items in the same bin.
We define anchor points as follows: $(0,0)$ for type $u$; all points with coordinates equal to $\frac{1+\epsilon}{2}$ or $\frac{1+\epsilon}{2}-(1+\epsilon)y_h$ for type $v$; all points with coordinates equal to $i\frac{(1+\epsilon)y_h}{2K}$ for type $t$.
As before, the anchor points for $u$ and $v$-items are a subset of the $t$-items' anchor points, and so we can pack one $u$-item together with $2^d-1$ $v$-items and $M$ $t$-items if we choose $M=\left( \frac{2K-y_{h}}{y_{h}} \right)^d - \left(\frac{K}{y_{h}}\right)^d - (2^d-1)\left( 2K \right)^d$.

For this input, any Harmonic-type algorithm uses at least $N$ bins for $u$-items, $N\frac{2^d-1}{2^d}=N(1-\frac{1}{2^d})$ bins for $v$-items and $\frac{NM}{\left(\frac{2K-y_h}{y_h}\right)^d}$ bins for $t$-items. 
This gives in total
\begin{align*}
& N\left( 2 - \frac{1}{2^d} + 1 - \left(\frac{K}{2K-y_h}\right)^d - (2^d-1)\left(\frac{2Ky_h}{2K-y_h}\right)^d \right)\\
& \xrightarrow{K\rightarrow \infty} N\left( 3 - \frac{1}{2^{d-1}} - (2^d-1)y_h^d \right)
\end{align*}
bins. We therefore can derive the following lower bound on the performance ratio:
\begin{equation}
R_\mathcal{A} \ge 3 - 1/2^{d-1} - (2^d-1)y_h^d \label{eq:ineq3}
\end{equation}

\subsection{Combined Lower Bound}

Given a certain set of parameters ($y_j$ and $m_j$), the maximum of the three right sides of inequalities (\ref{eq:ineq1}), (\ref{eq:ineq2}) and (\ref{eq:ineq3}) give us a bound on the competitive ratio of any Harmonic-type algorithm with this set of parameters. In order to get a general (worst-case) lower bound on $R_\mathcal{A}$, we need to find the minimum of this maximum over all possible sets of parameters.

This lower bound for $R_\mathcal{A}$ is obtained when equality holds in all of the inequalities (\ref{eq:ineq1}), (\ref{eq:ineq2}) and (\ref{eq:ineq3}). To see this, consider the following: We have $2h+1$ variables and $2h+1$ constraints. For $j \in \{1,\ldots, h\}$, we see that (\ref{eq:ineq1}) is increasing in $m_j$ and (\ref{eq:ineq2}) is decreasing in $m_j$. Next, let $c \in \{1,\ldots, h-1\}$. We see that (\ref{eq:ineq1}) for $j=h-c \in \{1,\ldots,h-1\}$ is decreasing in $y_c$, and (\ref{eq:ineq2}) for $j=h-c+1\in\{2,\ldots,h\}$ is increasing in $y_c$. Finally, we have that (\ref{eq:ineq2}) for $j=1$ is increasing in $y_h$ and (\ref{eq:ineq3}) is decreasing in $y_h$. This means, given certain parameters $y_j$ and $m_j$, if e.g. (\ref{eq:ineq3}) gives a smaller lower bound on $R_\mathcal{A}$ than (\ref{eq:ineq2}) with $j=1$ does, we can decrease the value of $y_h$ such that the maximum of the three lower bounds becomes smaller.

Setting the right hand side of (\ref{eq:ineq1}) equal to the right hand side of (\ref{eq:ineq2}), gives us $\frac{1}{m_j}=(1-y_{h-j+1})^d-\frac{1}{2^d}$ or alternatively $\frac{1}{m_{h-j+1}}=(1-y_j)^d-\frac{1}{2^d}$.
Plugging this into (\ref{eq:ineq1}) (replacing $j$ by $h-j+1$), we find that
\begin{align}
y_j = 1-\left( \frac{-2^d R_\mathcal{A} + 2^dy_{j-1}^d-4^dy_{j-1}^d-1+3\cdot 2^d-1/2^d}{2^d-1} \right)^{1/d} \label{eq:recy}
\end{align}
Recall that we require $1/3 = y_0 < y_1$. 
From this, combined with  (\ref{eq:recy}) for $j=1$, we obtain that
\begin{align*}
R_\mathcal{A} \ge 3 - 2\frac{2^d-1}{3^d}-\frac{2^d+1}{4^d}
\end{align*}
We list some values of the lower bound for several values of $d$ in Table \ref{tab:lower_bounds}.
\begin{table}[h]
\def\arraystretch{1.3}
	\setlength\tabcolsep{5px}
	\centering
	\begin{tabular}{|c|c|c|c|c|c|c|c|}
	\hline $d=$ & 1 & 2 & 3 & 4 & 5 & 6 & $\infty$ \\ 
	\hline $R_\mathcal{A}>$ & 1.58333 & 2.02083 & 2.34085 & 2.56322 & 2.71262 & 2.81129 & 3 \\ 
	\hline 
	\end{tabular} 
\caption{Lower bounds for Harmonic-type algorithms in dimensions 1 to 6 and limit for $d \rightarrow \infty$.}
\label{tab:lower_bounds}
\end{table}

Note that for $d=1$, our formula yields the bound of Ramanan et al. \cite{RBLL89}.
Surprisingly, it does not seem to help to analyze the values of $y_2, \ldots, y_h$. 
Especially, equations involving $y_j$ for $j > 1$ become quite messy due to the recursive nature of (\ref{eq:recy}). 
If $h$ is a very small constant like 1 or 2, we can derive better lower bounds for $R_\mathcal{A}$. 
For larger $h$, we can use the inequalities $y_1 < y_h, y_2 < y_h, y_3 < y_h$ (i.e. assuming that $h>3$) to derive \emph{upper} bounds on the best value $R_\mathcal{A}$ that could possibly be proven using this technique. 
These upper bounds are very close to 2.02 and suggest that for larger $h$, an algorithm in the class \class{} could come very close to achieving a ratio of 2.02 for these inputs. 
However, since the inequalities become very unwieldy, we do not prove this formally.

\begin{theorem}
No Harmonic-type algorithm for two-dimensional online hypercube packing can achieve an asymptotic performance ratio better than $2.0208$.
\end{theorem}

\section{Further Lower Bounds}

Inspired by \cite{HS15}, one could try to improve online algorithms for packing 2-dimensional squares by incorporating two ideas from the one-dimensional case: combining large items (i.e. items larger than $1/2$) and medium items (i.e. items with size in $(1/3,1/2]$) whenever they fit together (ignoring their type), and postponing the coloring decision. The former is intuitive, while the idea of the latter would be the following: When items of a certain type arrive, we first give them provisional colors and pack them into separate bins (i.e. one item per bin). After several items of this type arrived, we choose the smallest of them to be red and all others are colored blue. With following items of this type, we fill up the bins with additional items. However, simply adding two more red items to the bin with a single red item might be problematic: When filling up the red bins with two more red items, it could happen that these later red items are larger than the first one - negating the advantage of having the first red item be relatively small. Alternatively, we could leave the red item alone in its bin. This way, we make sure that at most $3/4$ of the blue items of a certain medium type are smaller than the smallest red item of this type, but we have more wasted space in this bin.

For both approaches discussed above we will show lower bounds on the competitive ratio that are only slightly lower or even higher than the lower bound established in Section \ref{sect:harmonic} for Harmonic-type algorithms.

\subsection{Always combining large and medium items}
\label{subsec:lb2}

First, we consider algorithms that combine small and large items whenever they fit together. We define a class of algorithms $B_1$ that distinguish, possibly among others, the following disjoint subintervals (types): 

\begin{itemize}
\item $I_m = (1/3, y]$ for some $y \in (1/3,1/2]$
\item $I_\lambda = (0,\lambda]$
\end{itemize}

These algorithms satisfy the following rules:

\begin{enumerate}
\item There is a parameter $\alpha$ s.t. an $\alpha$-fraction of the items of side length in $I_m$ are packed 3 per bin (``red items''), the rest are packed 4 per bin (``blue items'').
\item No bin that contains an item of side length in $I_\lambda$ contains an item of side length larger than $1/2$ or an item of side length in $I_m$.
\item Items of type $I_m$ are packed without regard to their size.
\end{enumerate}

Let $a,b\in I_m, a<b$. We consider two different inputs, both starting with the same set of items: $\frac{\alpha}{3}N$ items of size $b$ and $(1-\alpha/3)N$ items of size $a$ (i.e. in total $N$ items of size $a$ and $b$). By rule 3, the adversary knows beforehand which item will be packed in which bin, as they belong to the same type. Hence, the adversary can order these items in such a way that the items colored blue by the algorithm are all $a$-items, and in each bin with red items, there are two $a$- and one $b$-item. By rule 1, the online algorithm uses $(\frac{\alpha}{3}+\frac{1-\alpha}{4})N = \frac{3+\alpha}{12}N$ bins for items of this type.

The sizes $a$ and $b$ will tend towards $1/3$, as this way the adversary can maximize the total volume of sand (infinitesimally small items) that can be added to any bin in the optimal solution while not changing the way the algorithm packs these items and increasing the number of bins the algorithm needs for packing the sand items. Therefore, we will assume that $a$ and $b$ are arbitrarily close to $1/3$.

In the first input, after these medium items, $\frac{(1-\alpha/3)N}{3}$ items of size $1-a$ arrive, followed by sand of total volume $\frac{24+7\alpha}{324}N$. In the optimal solution, we can pack $\frac{\alpha}{12}N$ bins with four $b$-items and sand of volume $5/9$ each, and $\frac{(1-\alpha/3)N}{3}$ bins with three $a$-items, one $(1-a)$-item and sand of volume $\frac{\alpha}{12}N\cdot \frac{2}{9}$ each. Hence, the optimal solution uses $\frac{\alpha}{12}N + \frac{(1-\alpha/3)N}{3} = \frac{12-\alpha}{36}N$ bins.

The algorithm, however, cannot pack a large item into any of the bins with red medium items, as these always contain a $b$-item. Hence, in addition to the $\frac{3+\alpha}{12}N$ bins for medium items, the algorithm needs $\frac{(1-\alpha/3)N}{3}$ bins for large items and at least $\frac{24+7\alpha}{324}N$ bins for sand. This gives in total at least $\frac{213-2\alpha}{324}N$ bins, and a competitive ratio of at least
\begin{equation}
\frac{\frac{213-2\alpha}{324}N}{\frac{12-\alpha}{36}N}=\frac{213-2\alpha}{9(12-\alpha)}\label{eq:ineq4}
\end{equation}

In the second input, after the medium items, $N/3$ items of size $1/2+\epsilon$ will arrive, followed by sand of total volume $\frac{5}{36}N$. The optimal solution packs all medium items three per bin, using $N/3$ bins, and adds one large item and sand of volume $15/36$ in each such bin. In the algorithm's solution, large items can only be added to the $\alpha N/3$ bins containing three red items, i.e. it needs additional $N/3-\alpha N/3$ bins for the remaining $N/3-\alpha N/3$ large items. Finally, the algorithm uses at least $5/36 N$ bins for sand. The algorithm therefore uses in total at least $\frac{3+\alpha}{12}N + (1-\alpha)N/3 + 5/36 N = \frac{26-9\alpha}{36}N$ bins. This gives a competitive ratio of at least 
\begin{equation}
\frac{3(26-9\alpha)}{36} = \frac{26-9\alpha}{12}\label{eq:ineq5}
\end{equation}

Observe that (\ref{eq:ineq4}) is increasing in $\alpha$, while (\ref{eq:ineq5}) is decreasing in $\alpha$. Hence, the minimum over the maximum of the two bounds is obtained for the $\alpha$-value that makes both bounds equal, which is $\alpha=\frac{197-\sqrt{36541}}{27}\approx 0.2164$. For this $\alpha$, both bounds become larger than $2.0043$.

\begin{theorem}
No algorithm in class $B_1$ for two-dimensional online hypercube packing can achieve a competitive ratio of less than $2.0043$.
\end{theorem}

\subsection{Packing red medium items one per bin, postponing the coloring}

Now, consider the algorithm that packs red items alone into bins and makes sure that at most $3/4$ of the blue items of a certain type are smaller than the smallest red item of this type. We define a new class of algorithms $B_2$ that distinguish, possibly among others, the following disjoint subintervals (types): 

\begin{itemize}
\item $I_m = (1/3, y]$
\item $I_\lambda = (0,\lambda]$
\end{itemize}

Furthermore, algorithms in $B_2$ satisfy the following rules:

\begin{enumerate}
\item There is a parameter $\alpha$ s.t. an $\alpha$-fraction of the items of side length in $I_m$ are packed 1 per bin (``red items''), the rest are packed 4 per bin (``blue items'').
\item No bin that contains an item of side length in $I_\lambda$ contains an item of side length larger than $1/2$ or an item of side length in $I_m$.
\item Items of side length in $I_m$ are initially packed one per bin. At some regular intervals, the algorithm fixes some of these items to be red, and does not pack additional items of the same type witht them.
\end{enumerate}

From rule 3 we can conclude that the algorithm gives the following guarantee: $3/4$ of the blue items with size in $I_m$ are not smaller than the smallest red item with size in $I_m$.

Let $a,b\in I_m, a<b$ as before. We again consider two different inputs, both starting with the same set of items: $\alpha N+\frac{1-\alpha}{4}N$ items of size $b$, and $\frac{3(1-\alpha)}{4}N$ items of size $a$. They arrive in such an order that all red items are $b$-items, and all bins with blue items contain one $b$- and three $a$-items. We require the $b$-item in the blue bins because of the postponement of the coloring: If the first blue item in a bin was an $a$-item, the algorithm would choose this item to become red and not one of the $b$-items. By rule 1, the algorithm needs $\frac{1-\alpha}{4}N+\alpha N = \frac{1+3\alpha}{4}N$ bins for these $N$ items.

In the first input, after the medium items arrived, we get $\frac{1-\alpha}{4}N$ large items of size $1-a$, followed by sand of total volume $\frac{13+7\alpha}{144}N$. The optimal solution can pack the $a$-items three per bin together with one $(1-a)$-item, using $\frac{1-\alpha}{4}N$ bins for these items. The $b$-items are packed four per bin, using $(\frac{\alpha}{4}+\frac{1-\alpha}{16})N$ bins. Note that the empty volume in all bins of these two types is $\frac{1-\alpha}{4}N\cdot \frac{2}{9} + (\frac{\alpha}{4}+\frac{1-\alpha}{16})N\cdot \frac{5}{9} = \frac{13+7\alpha}{144}N$, i.e. it equals exactly the volume of the sand, so the sand can be filled in these holes without using further bins. Hence, the optimal number of bins is $\frac{1-\alpha}{4}N + (\frac{\alpha}{4}+\frac{1-\alpha}{16})N = \frac{5-\alpha}{16}N$.

The algorithm uses, as discussed before, $\frac{1+3\alpha}{4}N$ bins for the medium items of size $a$ and $b$. The large items cannot be added to red medium items, as they do not fit together, thus the algorithm uses $\frac{1-\alpha}{4}N$ additional bins for the large items. Finally, according to rule 2, at least $\frac{13+7\alpha}{144}N$ additional bins are needed to pack the sand. This gives in total at least $\frac{1+3\alpha}{4}N + \frac{1-\alpha}{4}N + \frac{13+7\alpha}{144}N = \frac{85+79\alpha}{144}N$ bins. We find that the competitive ratio is at least
\begin{equation}
\frac{\frac{85+79\alpha}{144}N}{\frac{5-\alpha}{16}N} = \frac{85+79\alpha}{9(5-\alpha)}\label{eq:ineq6}
\end{equation} 

In the second input, $N/3$ items of size $1/2+\epsilon$ arrive after the medium items, followed by sand of total volume $5/36N$. The algorithm packs this input the same way as a $B_1$ algorithm, so the analysis carries over. We get a competitive ratio of at least
\begin{equation}
\frac{\frac{26-9\alpha}{36}N}{N/3} = \frac{26-9\alpha}{12} \label{eq:ineq7}
\end{equation}

It can be seen that (\ref{eq:ineq6}) is a function increasing in $\alpha$, while (\ref{eq:ineq7}) is decreasing in $\alpha$, hence the minimum over the maximum of two bounds is reached when they are equal. In that case, $\alpha=\frac{529-\sqrt{274441}}{54}\approx 0.0950$, and the lower bound for the competitive ratio becomes larger than $2.0954$.

\begin{theorem}
No algorithm in class $B_2$ for two-dimensional online hypercube packing can achieve a competitive ratio of less than $2.0954$.
\end{theorem}

Note here that this is an even higher lower bound than the one shown in the previous Subsection \ref{subsec:lb2}, although we use postponement of the coloring here. This indicates that the space we waste by packing red medium items separately outweighs the advantage we get by having a guarantee about the size of the red item.

\section{Conclusion}

We have given improved general lower bounds as well as lower bounds for an important subclass of algorithms. We believe that our lower bound of 1.859 for rectangle packing could be further improved to 1.907 by extending the input sequence as given in Table \ref{tab:input_rect2}. However, we do not have a formal proof of what the
heaviest patterns are for the sets $T_i$ (our conjectures are listed in Table \ref{tab:patterns_for_dual_rect2}).
%
%
%
\begin{table}
\begin{center}
\begin{tabular}{|c|c|c|c|}
\hline 
type $j$ & width $w_j$ & height $h_j$ & $OPT(L_1\ldots L_j)/n\cdot 7224$\\ 
\hline 
1 & $1/4-30000\delta$ & \multirow{3}{*}{$1/1807+\epsilon$} & 1\\ 
\cline{1-2}\cline{4-4} 
2 & $1/4+10000\delta$ &  & 2 \\ 
\cline{1-2}\cline{4-4}
3 & $1/2+20000\delta$ &  & 4 \\ 
\hline 
4 & $1/4-3000\delta$ & \multirow{3}{*}{$1/43+\epsilon$} & 46 \\ 
\cline{1-2}\cline{4-4}
5 & $1/4+1000\delta$ &  & 88 \\ 
\cline{1-2}\cline{4-4}
6 & $1/2+2000\delta$ &  & 172 \\ 
\hline 
7 & $1/4-300\delta$ & \multirow{3}{*}{$1/7+\epsilon$} & 430 \\ 
\cline{1-2}\cline{4-4}
8 & $1/4+100\delta$ &  & 688 \\ 
\cline{1-2}\cline{4-4}
9 & $1/2+200\delta$ &  & 1204 \\ 
\hline 
10 & $1/4-30\delta$ & \multirow{3}{*}{$1/3+\epsilon$} & 1806 \\ 
\cline{1-2}\cline{4-4}
11 & $1/4+10\delta$ &  & 2408 \\ 
\cline{1-2}\cline{4-4}
12 & $1/2+20\delta$ &  & 3612 \\ 
\hline 
13 & $1/4-3\delta$ & \multirow{3}{*}{$1/2+\epsilon$} & 4515 \\ 
\cline{1-2}\cline{4-4}
14 & $1/4+\delta$ &  & 5418 \\ 
\cline{1-2}\cline{4-4}
15 & $1/2+2\delta$ &  & 7224 \\ 
\hline 
\end{tabular} 
\end{center}
\caption{The input items for a lower bound of 1.907.}
\label{tab:input_rect2}
\end{table}
\begin{table}
\centering
\begin{tabular}{|c|c|c|c|}
\hline 
$j$ & types to consider & heaviest patterns $p$ from $T_j$ & $w(p)\cdot 516211/516$ \\ 
\hline 
1 & 1 & $7224\times s_1$ & 7224 \\ 
\hline 
2 & 2, 4 & $5418 \times s_2$ & 5418 \\ 
\hline
\multirow{2}{*}3 & \multirow{2}{*}{3,4} & $1806\times s_3, 43\times s_4$ & \multirow{2}{*}{4816} \\ 
\cline{3-3}
&& $84\times s_3, 166\times s_4$&\\
\hline 
4 & 4 & $168\times s_4$ & 4704\\
\hline
5 & 5, 7 & $126\times s_5$ & 3528\\
\hline
\multirow{2}{*}{6} & \multirow{2}{*}{6, 7} & $42\times s_6, 7\times s_7$ & \multirow{2}{*}{3136}\\
\cline{3-3}
&& $12\times s_6, 22\times s_7$ &\\
\hline
7 & 7 & $24\times s_7 $ & 2688\\
\hline
\multirow{2}{*}{8} & \multirow{2}{*}{8, 10} & $18\times s_8$ & \multirow{2}{*}{2016}\\
\cline{3-3}
&& $6\times s_8, 8 \times s_{10}$&\\
\hline
9 & 9, 10 & $4\times s_9, 6\times s_{10}$ & 1904\\
\hline
10 & 10 & $8\times s_{10}$ & 1344\\
\hline
11 & 11, 13 & $3\times s_{11}, 4\times s_{13}$ & 1176\\
\hline
\multirow{2}{*}{12} & \multirow{2}{*}{12,13} & $2\times s_{12}, 2\times s_{13}$ & \multirow{2}{*}{1008}\\
\cline{3-3}
&& $1\times s_{12}, 4\times s_{13}$ &\\
\hline
13 & 13 & $4\times s_{13}$&672\\
\hline
14&14&$3\times s_{14}$&504\\
\hline
15&15&$1\times s_{15}$&336\\
\hline
\end{tabular} 
\caption{We believe that these are the patterns that need to be considered for verifying the feasibility of the dual solution of the 1.907 lower bound.}
\label{tab:patterns_for_dual_rect2}
\end{table}

The main open question is how to reduce the gap between the upper and lower bounds for these problems, which remain fairly large.

\bibliographystyle{elsarticle-num} 
\bibliography{biblio}

\end{document}